%% file: root.tex
\documentclass[letterpaper, 10 pt, conference]{ieeeconf}  

\IEEEoverridecommandlockouts                              

\usepackage{mathptmx} 

\usepackage{times} 

\usepackage{amsthm}
\usepackage{multicol}
\usepackage{cite}
\newtheorem{definition}{Definition}
\newtheorem{corollary}{Corollary}
\newtheorem{remark}{Remark}
\usepackage[bookmarks=true]{hyperref}
\input{packages.tex}

\input{commands.tex}
\title{\LARGE \bf
High-Relative Degree Stochastic Control Lyapunov and Barrier Functions 
}

\author{Meenakshi Sarkar$^{1}$, Debasish Ghose$^{2}$ and Evangelos A. Theodorou$^{3}$
\thanks{$^{1}$Meenakshi Sarkar is with the Aerospace Engineering Department,
        Indian Institute of Science, Bangalore, India. This work was done during her stay as a visiting scholar at the Aerospace Engineering Department, Georgia Institute of Technology, USA 
        {\tt\small meenakshisar@iisc.ac.in}}%
\thanks{$^{2}$Debasish Ghose is with the faculty of Aerospace Engineering Department,
        Indian Institute of Science, Bangalore, India.
        {\tt\small dghose@iisc.ac.in}}%
\thanks{$^{3}$Evangelos A. Theodorou is with the faculty of Aerospace Engineering Department,
        Georgia Institute of Technology, USA 
        {\tt\small evangelos.theodorou@gatech.edu}}%
}

\begin{document}

\maketitle
\thispagestyle{empty}
\pagestyle{empty}

\begin{abstract}
We introduce High-Relative Degree Stochastic Control Lyapunov functions and Barrier Functions as a means to ensure asymptotic stability of the system and incorporate state dependent high relative degree safety constraints on a non-linear stochastic systems. Our proposed formulation also provides a generalisation to the existing literature on control Lyapunov and barrier functions for stochastic systems. The control policies are evaluated using a constrained quadratic program that is based on control Lyapunov and barrier functions. Our proposed control design is validated via simulated experiments on a relative degree 2 system (2 dimensional car navigation) and relative degree 4 system (two-link pendulum with elastic actuator).

\end{abstract}

\section{Introduction}
With the rapid advancement of autonomous systems in various sectors such as automobile, aviation, finance and medical, there has been a surge on research interest into safety verification of control systems. Safety plays a crucial role in various engineering application and thus researchers have investigated various methods such as barrier  methods \cite{nguyen2016exponential}, reachable sets \cite{althoffACM2011}, and discrete approximations \cite{mitra2013}, to ensure safety certification of the system. Recently \cite{ames2014control}, \cite{ames2019control} introduced \acp{CBF} which aims at designing controllers with verifiable safety bounds.  

While there has been a considerable amount of work on investigating barrier function methods based on the invariant set theorem \cite{ames2019control} for designing safe controllers for deterministic systems, similar literature on stochastic systems is rather scarce. Only recently, \cite{clark2019control} provided some sufficient conditions to satisfy safety constrains for stochastic systems. This work was limited in the sense that it was applicable only to scenarios where the constrained states are directly actuated. For systems where the observed states are not directly actuated via control gives rise to higher relative degree systems. For underactuated deterministic systems, \cite{nguyen2016exponential}  proposed exponential \ac{CBF} and demonstrated its effectiveness on systems with  relative degrees 6 and  4.  Recent work by \cite{xiao2019} tried to generalise the exponential barrier function approach and designed controller for an automatic cruise control problem. But currently there are no formation for a high relative degree control barrier function for stochastic systems. To the best of our knowledge, the present paper is the first attempt to generalise higher relative degree \ac{SCBF}. Interestingly, in our formulation, the recent work on \ac{SCLF} by \cite{clark2019control} turns out to be a special case where the relative degree of the system is one. 

We further investigate the links between the formulation of \ac{SCBF} and \ac{SCLF}. There has been a considerable amount work \cite{deng1997}, \cite{nishimura2013stochastic}, \cite{wang2009performance}, \cite{stochastic_stability_book} done on investigating the stability of stochastic systems in the Lyapunov framework. Recently, \cite{fan2019bayesian} investigated the CLF-CBF formulation along with a Bayesian deep learning framework to design adaptive controllers for systems operating in uncertain environments. While the existing literature provide concrete foundation for Lyapunov function based stochastic controllers, they still do not provide a generalised framework that one can apply to systems with high relative degree or underactuated systems.  While in the past the iLQG and dynamic programming framework \cite{todorov2005generalized} has been proven to be  effective in case of stochastic systems, iLQG inherently assumes the system to be stabilizable which might not be the case for a underactuated system. Incidentally, there have been very few attempts in the past to address stability of stochastic systems with high relative degree. 

In this paper we also provide a generalised Lyapunov function formulation for stochastic systems with high relative degree or high degree of underactuation. We also show how one can derive the existing formulation on \ac{SCLF} \cite{deng1997}, \cite{stochastic_stability_book} as a special case where the relative degree is equal to one.
 For our analysis we assume complete system state information. The contributions of the paper are as follows:
\begin{enumerate}
    \item We propose a generalised high relative degree \ac{SCBF} formulation  and derive sufficient conditions for the system to satisfy safety constraints with probability 1.
    \item We extend the high relative degree \ac{SCBF} formulation to a generalised high relative degree \ac{SCLF} formulation  and derive  sufficient conditions for asymptotic stability in the sense of probability.
    \item We show that the previous formulations on stochastic CLF-CBF was a special case of our generalised formulation where the relative degree is 1.
    \item We propose an approach to evaluate control policies  using a stochastic CLF-CBF based quadratic programs. The stability and performance of the controller in probabilistic sense can be guaranteed via stochastic CLFs. 
    \item We validate the proposed controller via simulated experiments on two different stochastic systems of high relative degree.
\end{enumerate}
The rest of the paper is organised as follows: Section \ref{sec:background} provides the mathematical preliminaries and  the problem formulation. Section \ref{sec:SCBF} derives the generalised higher order \ac{SCBF} and Section \ref{sec:SCLF} considers the higher order \ac{SCLF}. We provide the controller formulation in Section \ref{sec:control_policy} followed by  simulations results on two high relative degree stochastic systems with relative degree 2 (2D car navigation) and 4 (2-link pendulum with elastic actuator) in Section \ref{sec:results}. Finally we conclude the paper in Section \ref{sec:conclusion}.

\section{problem formulation and mathematical preliminaries}\label{sec:background}
We consider the stochastic dynamics system as:
\begin{equation}
    \rd\vx(t) = \big(\vf(\vx(t)) + \vG(\vx(t)) \vu(t) \big)\rd t +  \vSigma(\vx(t)) \rd \vw(t)
\label{eq:SDE}
\end{equation}
where, $\vw(t)$ is an $n_w$ dimensional standard Brownian motion; 
$\vx\in\Rb^{n_{x}}$, and $\vu\in \big(\Rb^{n_u}\cap \calU([0,T]\big)$ denote the state and control vectors, respectively, where $\calU([0,T])$ represents the set of admissible controls for a fixed finite time horizon $T\in[0,\infty)$ (i.e $0\leq t< T< \infty$). The functions $\vf: \Rb^{n_x} \rightarrow \Rb^{n_x}$, $\vG:\Rb^{n_x} \rightarrow \Rb^{n_x \times n_u}$ and $\vSigma: \Rb^{n_x} \rightarrow \Rb^{n_x \times n_w}$ represent the uncontrolled system dynamics, the actuator dynamics, and the diffusion matrices,  respectively.

The drift terms are assumed to be modelled as:
\begin{equation}
    \va(\vx,t)\rd t= \big(\vf(\vx(t)) + \vG(\vx(t)) \vu(t) \big)\rd t
\end{equation}
For the dynamical system defined by \eqref{eq:SDE} safety is ensured if  $\vx(t)\in \mathcal C_{safe}$ for all $t$, where the set $\mathcal C_{safe}$ denotes the safe region of operation and  is defined through a locally Lipschitz function $h: \Rb ^{n_x}\rightarrow \Rb$ \cite{ames2014control} as: 
\begin{align*}
    \mathcal{C}_{safe} &= \{\vx : h(\vx)\ge 0 \} \\
    \delta \mathcal{C}_{safe} &= \{\vx : h(\vx) = 0 \}.
\end{align*}
\begin{definition}\label{def:scbf_deg1}
We define an \ac{SCBF} $B(\vx): \mathbb{R}^{n_x}\rightarrow \mathbb{R}$ which is locally Lipschitz and twice differentiable on int$(\mathcal{C}_{safe})$, and satisfies the  property that
    there exist class-K functions $\alpha_1$,  $\alpha_2$, and $\alpha_3$, such that
    \begin{equation}\label{eq:cbf_cond1}
        \frac{1}{\alpha_1(h(\vx))}\leq B(\vx) \leq \frac{1}{\alpha_2(h(\vx))}, \, \forall \vx \in \mathcal C_{safe}
    \end{equation}
    \begin{equation}\label{eq:cbf_cond2}
        \frac{\partial B}{\partial \vx}\T (\vf + \vG\vu) + \frac{1}{2}\tr \Big(\frac{\partial^2 B}{\partial \vx^2}\vSigma\vSigma\T\Big) \leq \alpha_3(h(\vx)), \, \forall t
    \end{equation}
\end{definition}
\begin{theorem}\label{th:scbf_deg1}
If there exists a \ac{SCBF} of the form given in Definition \ref{def:scbf_deg1} for a process given in Eqn. (\ref{eq:SDE}), then $\text{Pr}(\vx_t\in \mathcal {C}_{safe})=1$ for $\forall t$ given that $x(t_0)\in \mathcal C_{safe}$. The detailed proof can be found in \cite{clark2019control}.
\end{theorem}
\section{High Relative Degree Stochastic Control Barrier Function }\label{sec:SCBF}
\begin{definition}
Consider a Stochastic Dynamical System defined by Eq. (\ref{eq:SDE}) with safety set defined by $\mathcal {C}_{safe}$. We define a locally Lipschitz and twice-differentiable control barrier function $B_0(\vx):$ $\mathbb{R}^{n_x}\rightarrow \mathbb{R}$  that satisfies Eqn. (\ref{eq:cbf_cond1}) and Eqn. (\ref{eq:cbf_cond2}) such that $\mathit{L}_g \mathit{L}_f^{r_b-1}B_0(\vx)u \neq 0$, and
$\mathit{L}_g \mathit{L}_f^{r_b-2}B_0(\vx)u=\mathit{L}_g \mathit{L}_f^{r_b-3}B_0(\vx)u = \cdots =\mathit{L}_g \mathit{L}_fB_0(\vx)u= 0$.
The function $B_0(\vx)$ is said to be Stochastic Control Barrier Function with relative degree $r_b$. The function $B(\vx)$ in definition \ref{def:scbf_deg1} is a special case where the relative degree of the control barrier function $r_b=1$.
\end{definition}
\begin{definition} 
We define a sequence of safe sets as:
\begin{align*}
    \mathcal{C}_0= \mathcal{C}_{safe} &= \{\vx : \psi_0(\vx)\ge 0 \} \\
    \delta \mathcal{C}_0= \delta \mathcal{C}_{safe} &= \{\vx : \psi_0(\vx) = 0 \}\\.
    \mathcal{C}_1 &= \{\vx : \psi_1(\vx)\ge 0 \} \\
    \delta \mathcal{C}_1 &= \{\vx : \psi_1(\vx) = 0 \}\\
    \vdots\\
    \mathcal{C}_{r_b} &= \{\vx : \psi_{r_b}(\vx)\ge 0 \} \\
    \delta \mathcal{C}_{r_b}= &= \{\vx : \psi_{r_b}(\vx) = 0 \}
\end{align*}
where $\psi_i$ with $i= 0,1,\cdots, r_b$ is defined as follows:
\begin{equation*}
\begin{split}
\psi_0(\vx) = &h(\vx)\\
    B_0(\vx) = &\frac{\gamma_0}{\psi_0(\vx)}\\
    \psi_1(\vx) =& \alpha_3^0(\psi_0(\vx))- \left( \frac{\partial B_0^T}{\partial \vx}\va(\vx)+\frac{1}{2}\tr \Big(\frac{\partial^2 B_0}{\partial \vx^2}\vSigma\vSigma\T\Big)\right)\\
    B_1(\vx) =& \frac{\gamma_1}{\psi_1(\vx)}\\
    \psi_2(\vx) =& \alpha_3^1(\psi_1(\vx))- \left( \frac{\partial B_1^T}{\partial \vx}\va(\vx)+\frac{1}{2}\tr \Big(\frac{\partial^2 B_1}{\partial \vx^2}\vSigma\vSigma\T\Big)\right)\\
    \vdots\\
    \psi_{r_b-1}(\vx)& = \alpha_3^{r_b-2}(\psi_{r_b-2}(\vx))-\Upsilon_{r_b-2}\\
    B_{r_b-1}(\vx) &= \frac{\gamma_{r_b-1}}{\psi_{r_b-1}(\vx)}\\
    \psi_{r_b}(\vx) = &\alpha_3^{r_b-1}(\psi_{r_b-1}(\vx))-\Upsilon_{r_b-1}\\
    \end{split}
\end{equation*}
where, $\Upsilon_j=\left( \frac{\partial B_{j}^T}{\partial \vx}\va(\vx)+\frac{1}{2}\tr \Big(\frac{\partial^2 B_{j}}{\partial \vx^2}\vSigma\vSigma\T\Big)\right)$; $\gamma_i\in \mathbb{R}^+$ and $\alpha_3^i$,  with $i\in \{0,1,\cdots, r_b-1\}$, are class K function; $B_i(\vx)$ are locally Lipschitz and twice differentiable on int$(\mathcal C_i)$ and designed such that:
\begin{align}
    \frac{1}{\alpha_1^i(\psi_i(\vx))}\leq B_i(\vx) \leq \frac{1}{\alpha_2^i(\psi_i(\vx))}, \, \forall \vx \in \mathcal {C}_i
\end{align}
Here, $\alpha_1^i$ and $\alpha_2^i$ for $i= 0,1,\cdots r_b-1$, are all class K functions. 
\end{definition}
\begin{theorem}\label{th:scbf}
If $x(t_0)\in \mathcal{C}_{r_b}$ and $\psi_{r_b}(\vx)\geq0$ then $\text{Pr}(\vx_t\in \mathcal {C}_{safe})=1$ for $\forall t$ where $r_b$ denotes the relative degree of the control barrier function.
\end{theorem}
\begin{proof}
The proof is done via induction. We onsider the case with $r_b=1$ we have
\begin{align}
    \psi_{1}(\vx)& \geq0
\end{align}
\begin{align}    
    \implies  \big( \frac{\partial B_0^T}{\partial \vx}\va(\vx)+\frac{1}{2}\tr \Big(\frac{\partial^2 B_{0}}{\partial \vx^2}\vSigma\vSigma\T\Big)\big)& \leq \alpha_3^0(\psi_0(\vx))\\
    \implies  \big( \frac{\partial B_0^T}{\partial \vx}\va(\vx)+\frac{1}{2}\tr \Big(\frac{\partial^2 B_{0}}{\partial \vx^2}\vSigma\vSigma\T\Big)\big)& \leq \alpha_3^0(h_0(\vx))
\end{align}
From Eq (8) along with the Eq.(5) we can say 
\begin{align*}
    \text{Pr}(\vx_t\in \mathcal {C}_{0})=1\\
    \implies \text{Pr}(\vx_t\in \mathcal {C}_{safe})=1
\end{align*}
This directly follows from the proof given in \cite{clark2019control}.

Now assume that the relative degree of the CBF is $j$. Now, from Theorem 1. we have 
\begin{align}
    \psi_{j}(\vx)& \geq0
\end{align}
\begin{align}
    \implies  \big( \frac{\partial B_{j-1}^T}{\partial \vx}\va(\vx)+\frac{1}{2}\tr \Big(\frac{\partial^2 B_{j-1}}{\partial \vx^2}\vSigma\vSigma\T\Big)\big)& \leq \alpha_3^0(\psi_{j-1}(\vx))
\end{align}

Following the proof given for stochastic CBF in \cite{clark2019control} we can derive that given Eq. 10 and Eq 5, 
\begin{align*}
    \text{Pr}(\vx_t\in \mathcal {C}_{j-1})=1 \text{ for } \forall t
\end{align*}
Now, given $\vx\in \mathcal {C}_{j-1} $ for $\forall t$ we can say that 
\begin{align}
    \psi_{j-1}(\vx)& \geq0
\end{align}
\begin{align}
    \implies  \big( \frac{\partial B_{j-2}^T}{\partial \vx}\va(\vx)+\frac{1}{2}\tr \Big(\frac{\partial^2 B_{j-2}}{\partial \vx^2}\vSigma\vSigma\T\Big)\big)& \leq \alpha_3^0(\psi_{j-2}(\vx))
\end{align}
With Eq. 12 and Eq. 5 we can again prove that:
\begin{align*}
    \text{Pr}(\vx_t\in \mathcal {C}_{j-2})=1 \text{ for } \forall t
\end{align*}
Thus, with the property of induction and  Eq 8 and Eq 5 we can iteratively prove that
\begin{align*}
    \text{Pr}(\vx_t\in \mathcal {C}_{j-2})&=1 \text{ for } \forall t\\
    \implies \text{Pr}(\vx_t\in \mathcal {C}_{j-3})&=1 \text{ for } \forall t\\
    \vdots\\
    \implies \text{Pr}(\vx_t\in \mathcal {C}_{1})&=1 \text{ for } \forall t\\
    \implies \text{Pr}(\vx_t\in \mathcal {C}_{0})&=1 \text{ for } \forall t\\
    \implies \text{Pr}(\vx_t\in \mathcal {C}_{safe})&=1 \text{ for } \forall t 
\end{align*}
\end{proof}
\begin{corollary}
$\mathcal{C}_{r_b}\subseteq \mathcal{C}_{r_b-1} \subseteq\cdots \subseteq \mathcal{C}_{1}\subseteq \mathcal{C}_{0}$
\end{corollary}
\begin{proof}
This can be easily derived from Theorem \ref{th:scbf}. 
\end{proof}
\begin{remark}
It can be noted from Theorem \ref{th:scbf} that  
\begin{align*}
    \frac{\partial B_{i}^T}{\partial \vx}\va(\vx)= \frac{\partial B_{i}^T}{\partial \vx}\vf(\vx) \text{ for } i= 0,1, \cdots, r_b-2
\end{align*}
and 
\begin{align*}
    \frac{\partial B_{r_b-1}^T}{\partial \vx}\va(\vx)= \frac{\partial B_{r_b-1}^T}{\partial \vx}(\vf(\vx)+\vG\vu)
\end{align*}

\end{remark}
\section{High Relative Degree Stochastic Control Lyapunov Function }\label{sec:SCLF}
\begin{definition}
\label{def:stable_probablity}
A solution $x(t)\equiv 0$ of equation \ref{eq:SDE} is said to be stable in \textit{probability} for $t\geq 0$ if for any $s\geq0$ and $\epsilon > 0$
\begin{equation*}
     \underset{\eta\rightarrow0}{\mathrm{lim}}\mathbf{P}\left \{ \underset{t>s}{\mathrm{sup}} \left | x^{s,\eta}(t) \right |>\epsilon\right \}=0
\end{equation*}
\end{definition}
\begin{theorem}\label{th:stable_sclf}
Let $U$ be a domain containing the line $x=0$ and there exist a smooth positive definite function $V(\vx)\in \mathbb{C}_2^0(U)$ s.t $V(\vx): \mathbb{R}^{n_x}\rightarrow \mathbb{R}_+$ that satisfies
\begin{equation*}
    \mathit{L}_{f}V+\mathit{L}_{G}Vu+\frac{1}{2}\textrm{tr}\{\frac{\partial^2 V}{\partial \vx^2}\Sigma \Sigma ^T\} \leq0
\end{equation*}
The function $V(\vx)$ is known as the Stochastic Control Lyapunov Function \ac{SCLF}. Detailed proof can be found in \cite{stochastic_stability_book}
\end{theorem}
\begin{definition}
\label{def:asymptotic_stable_probablity}
A solution $x(t)\equiv 0$ of equation \ref{eq:SDE} is said to be asymptotically stable if it is stable in probability and moreover 
\begin{equation*}
     \underset{\eta\rightarrow0}{\mathrm{lim}}\mathbf{P}\left \{ \underset{t\rightarrow \infty}{\mathrm{lim}}  x^{s,\eta}(t)=0\right \}=1
\end{equation*}
\end{definition}
\begin{theorem}\label{th:asymp_stable_sclf}
The system \ref{eq:SDE} is globally asymptotically stabilizable in probability if there exist a SCLF that satisfies 
\begin{equation*}
  \mathit{L}_{f}V+\mathit{L}_{G}Vu+\frac{1}{2}\textrm{tr}\{\frac{\partial^2 V}{\partial \vx^2}\Sigma \Sigma ^T\} <0
\end{equation*}
Detailed proof can be found in \cite{stochastic_stability_book} and \cite{deng1997}
\end{theorem}
\begin{definition}\label{def:higher_sclf}
Let there exist a smooth function $V_0(\vx)$ that satisfies the conditions given in theorem \ref{th:asymp_stable_sclf} such that, $\mathit{L}_G \mathit{L}_f^{r_l-1}V_0(\vx)u \neq 0$ and 
$\mathit{L}_G \mathit{L}_f^{r_b-2}V_0(\vx)u=\mathit{L}_G \mathit{L}_f^{r_l-3}V(\vx)u = \cdots =\mathit{L}_G \mathit{L}_fV_0(\vx)u= 0$. The function $V(\vx)$ is said to be Stochastic Control Lyapunov Function with relative degree $r_l$. \\
The function $V(\vx)$ defined in theorem \ref{th:asymp_stable_sclf} is a special case where the relative degree of the \ac{SCLF}, $r_l=1$
\end{definition}
\begin{definition}
Define a sequence of Lyapunov sets as:
\begin{align*}
    \mathcal{L}_1= \mathcal{L}_{stable} &= \{\vx : \chi_{1}(\vx)\ge 0 \} \\
    \delta \mathcal{L}_1= \delta \mathcal{L}_{stable} &= \{\vx : \chi_1(\vx) = 0 \}\\
    \mathcal{L}_2 &= \{\vx : \chi_2(\vx)\ge 0 \} \\
    \delta \mathcal{L}_2 &= \{\vx : \chi_2(\vx) = 0 \}\\
    \vdots\\
    \mathcal{L}_{r_l} &= \{\vx : \chi_{r_l}(\vx)\ge 0 \} \\
    \delta \mathcal{L}_{r_l}= &= \{\vx : \chi_{r_l}(\vx) = 0 \}
\end{align*}
where, we define  $\chi_i$ with $i= 0,1,\cdots, r_l$  as follows:
\begin{align*}
    g(\vx)= &d- \big(\mathit{L}_{f}V_0+\mathit{L}_{G}V_0u+\frac{1}{2}\textrm{tr}\{\frac{\partial^2 V_0}{\partial \vx^2}\Sigma \Sigma ^T\} \big)\\
    \chi_1(\vx) = &g(\vx)\\
    V_1(\vx) = &\frac{\upsilon_1}{\chi_1(\vx)}\\
    \chi_2(\vx) =& \beta_3^1(\chi_1(\vx))- \left( \frac{\partial V_1^T}{\partial \vx}\va(\vx)+\frac{1}{2}\tr \Big(\frac{\partial^2 V_1}{\partial \vx^2}\vSigma\vSigma\T\Big)\right)\\
    V_2(\vx) =& \frac{\upsilon_2}{\chi_2(\vx)}\\
    \chi_3(\vx) =& \beta_3^2(\chi_2(\vx))- \left( \frac{\partial V_2^T}{\partial \vx}\va(\vx)+\frac{1}{2}\tr \Big(\frac{\partial^2 V_2}{\partial \vx^2}\vSigma\vSigma\T\Big)\right)\\
    \vdots\\
    \chi_{r_l-1}(\vx)& = \beta_3^{r_l-2}(\chi_{r_l-2}(\vx))-\Zeta_{r_l-2}\\
    V_{r_l-1}(\vx) &= \frac{\upsilon_{r_b-1}}{\chi_{r_b-1}(\vx)}\\
    \chi_{r_l}(\vx) = &\beta_3^{r_l-1}(\chi_{r_b-1}(\vx))-\Zeta_{r_l-1}\\
\end{align*}
where, $\Zeta_j=\left( \frac{\partial V_{j}^T}{\partial \vx}\va(\vx)+\frac{1}{2}\tr \Big(\frac{\partial^2 V_{j}}{\partial \vx^2}\vSigma\vSigma\T\Big)\right)$ and $d\in \mathbb{R}$ is the relaxation variable. $\upsilon_i\in \mathbb{R}^+$ and $\beta_3^i$s are class K functions where $i\in \{1,\cdots, r_l-1\}$ . $V_i(\vx)$ are are locally Lipschitz, twice differentiable on int$(\mathcal L_i)$ and designed such that:
\begin{align}
    \frac{1}{\beta_1^i(\chi_i(\vx))}\leq V_i(\vx) \leq \frac{1}{\beta_2^i(\chi_i(\vx))}, \, \forall \vx \in \mathcal {L}_i
\end{align}
Here $\beta_1^i$ and $\beta_2^i$ for $i= 1,\cdots r_b-1$ are all class K functions. 
\end{definition}
\begin{theorem}\label{th:sclf}
If $x(t_0)\in \mathcal{L}_{r_l}$ and $\chi_{r_l}(\vx)\geq0$, then the system governed with equation \ref{eq:SDE} is globally asymptotically stabilizable in probability.
\end{theorem}

\begin{proof}\label{pf:sclf}
From theorem \ref{th:scbf} we can show that given,  $\chi_{r_l}(\vx)\geq0$,
\begin{align}
 \text{Pr}(\vx_t\in \mathcal {L}_{stable})=1\text{ for }\forall t, 
\end{align}
\\
Using theorem \ref{th:asymp_stable_sclf} we can easily show that system will be globally asymptotically stabilizable in probability. 
\end{proof}
\begin{corollary}
$\mathcal{L}_{r_l}\subseteq \mathcal{L}_{r_l-1} \subseteq\cdots \subseteq \mathcal{L}_{2}\subseteq \mathcal{L}_{1}$
\end{corollary}
\begin{proof}
This can be easily derived from Theorem \ref{th:sclf}. 
\end{proof}
\begin{remark}
It can be noted from Theorem \ref{th:sclf} that  
\begin{align*}
    \frac{\partial V_{i}^T}{\partial \vx}\va(\vx)= \frac{\partial V_{i}^T}{\partial \vx}\vf(\vx) \text{ for } i= 0,1, \cdots, r_b-2
\end{align*}
and 
\begin{align*}
    \frac{\partial V_{r_l-1}^T}{\partial \vx}\va(\vx)= \frac{\partial V_{r_l-1}^T}{\partial \vx}(\vf(\vx)+\vG\vu)
\end{align*}
\end{remark}
\section{Control Policies from \ac{SCLF}-\ac{SCBF} based Quadratic Programs}\label{sec:control_policy}
In this section we discuss how the control policies are formulated in the form of a \ac{QP} based on \ac{HDSCLF}  and \ac{CBF}. The details of \ac{HDSCLF}-\ac{CBF} based \ac{QP} are as follows:  \\
\noindent\rule{\linewidth}{0.4pt}\vspace{1mm}
\par \textbf{HDSCLF-CBF-QP}:
\begin{align}
    \vu^* =& \underset{\vu, d}{\mathrm{ argmin}} \text{    }& &\vu^T Q\vu+ pd^2 \label{eq:qp}\\
    &\text{  s.t. } & &\chi_{r_l}(\vx)\geq0 \tag{\textbf{SCLF}}\label{eq:SCLF}\\
    &\text{ } & &\psi_{r_b}(\vx)\geq0 \tag{\textbf{SCBF}}\label{eq:SCBF}\\
    &\text{ } & &\vu_{upper}\geq \vu \geq \vu_{lower} \tag{\textbf{Constraint}}\label{eq:constraint}
\end{align}
\noindent\rule{\linewidth}{0.4pt} \vspace{2mm}
Equation \ref{eq:qp} denotes the cost of \ac{QP} subjected to the \ac{SCLF} and \ac{SCBF} conditions and the higher and lower limits of the control effort.
\section{Simulation Results}\label{sec:results}
In order to validate our proposed method we consider two tasks: 2D car navigation in the presence of multiple obstacle (non-linear system with relative degree 2) and swinging a 2-link pendulum with elastic actuator (relative degree 4 system with 2 degree underactuation). We assume the noise only enters through the control channels similar to \cite{pereira2019learning} and \cite{exarchos2018stochastic}. The details of the simulated experiments and results are as follows:
\subsection{Two-dimensional Car navigation (Relative Degree 2)}
We validate our proposed method on a relatively complex and nonlinear scenarios of a two-dimensional car navigation. We consider a simple 2D circle following a simplified car dynamics as given in \cite{xie2017} for a time horizon $T= 8$s :
\begin{align*}
   \dot{x}=v sin(\theta), & &
   \dot{y}=v cos(\theta)\\
   \dot{\theta}= vu^{\theta},& &
   \dot{v}= u^{v}
\end{align*}
where $\vx=[x, y, \theta, v]$ denotes the state of the car which includes 2D position, orientation and forward velocity. Control variables $u^\theta$ and $u^v$ changes the steering angle and forward velocity respectively. Here we want to find control policies to navigate the car from the initial x-y position of  $[0,0]$ with orientation of $[0,0]$ to the final goal position $[x_{desired},y_{desired}]= [4,4]$ in the presence of single and multiple obstacles. 
The 2D constraints on the positional states from the obstacle avoidance problem leads to a relative degree 2 control barrier function as the positional states are not directly actuated. The controller also tries to reach the target position while avoiding the obstacles which leads to a relative degree 2 control Lyapunov function formulation. The control Lyapunov function was defined as :
\begin{equation*}
    V_0(\vx)= \frac{1}{2}\big((x-x_{desired})^2+(y-y_{desired})^2\big)
\end{equation*}
The barrier function and the safe sets for each obstacle $i$ are defined as follows:
\begin{align*}
   \psi_0^i(\vx) &= h^i(\vx)= \big((x^i-x^i_c)^2+ (y^i-x^i_c)^2- (r^i)^2 \big)\\
   \mathcal{C}_{safe}^i &= \{\vx : \psi_0^i(\vx)\ge 0 \}\\
    B_0^i(\vx)&= \frac{\gamma_0^i}{\psi_0^i(\vx)}
\end{align*}
and 
\begin{equation*}
\mathcal{C}_{safe}= \bigcup_{i=1}^{n_o}\mathcal{C}_{safe}^i
\end{equation*}
 Where $(x^i_c, y^i_c)$ and $r^i$ denote the 2D position of the centre and radius of the i$^{th}$ obstacle respectively. $n_o$ represents the no of obstacle present. \\
 
 \begin{figure}[H]
  \includegraphics[width=\linewidth, ]{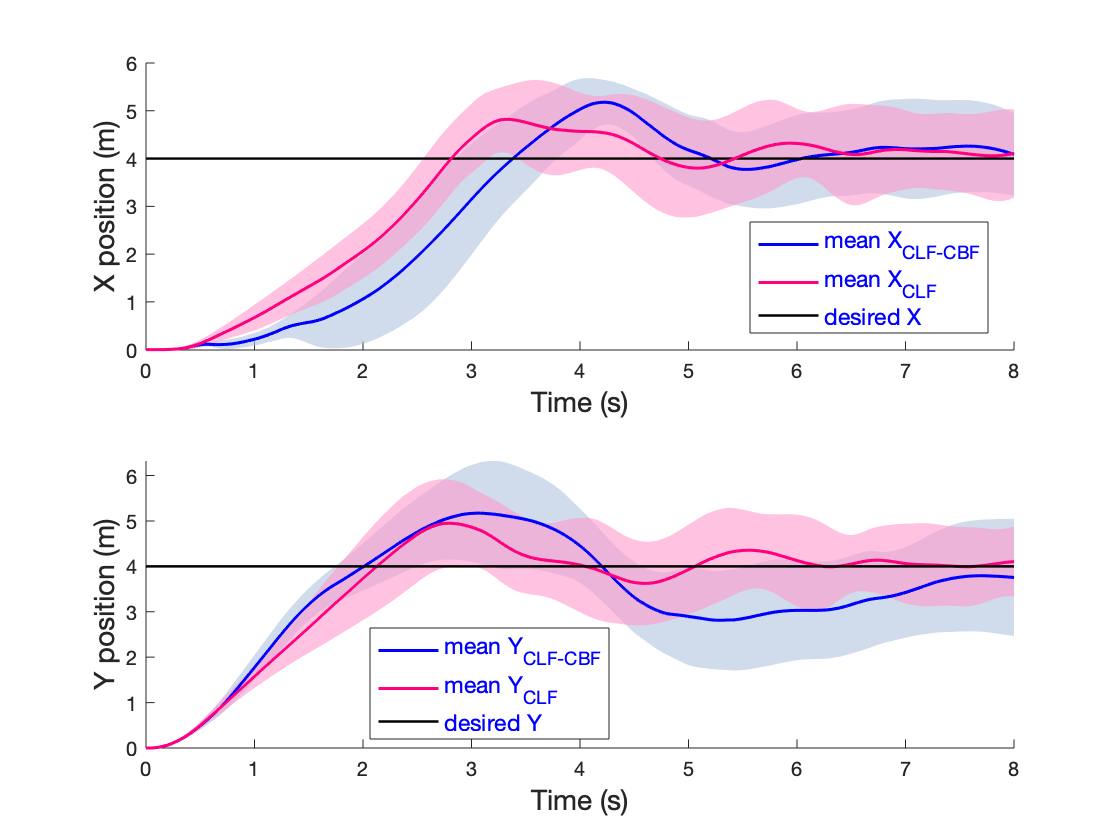}
  \caption{Multiple obstacle avoidance Task: The mean and the sigma distribution of the $x$ and $y$ state trajectories for the CLF-CBF controller and the CLF controller are plotted here in sheds of blue and green respectively.   }
\label{fig:2DCar_xy_multi_obstacle}
\end{figure}
 For our simulations of single obstacle the centre and radius was taken as $(3,2.5)$m and 0.6m respectively. For the multiple obstacle case we took 3 obstacles placed at $(1,1)$m, $(1,4)$m and $(3,2.5)$m with radius of 0.4m,0.4m and 0.6m respectively.
 The diagonal elements of the $Q$ matrix are chosen as $[1000,10]$ and $p=1000$. We simulated 20 sample trajectories for the experiments and the multiple obstacle avoidance task are shown in figure \ref{fig:2DCar_xy_multi_obstacle} and figure \ref{fig:2DCar_traj_multi_obstacle}. Interestingly from the enlarged portion of figure \ref{fig:2DCar_traj_multi_obstacle} we can see for the HDSCLF-CBF controller the trajectories take a \textit{U-turn} after encountering the obstacle and moves towards the goal point. The circular trajectories around the goal point in figure \ref{fig:2DCar_traj_multi_obstacle} is expected due to the asymptotic nature of the stability in the sense of probability from theorem \ref{th:sclf}.

\begin{figure}[H]
  \includegraphics[width=0.97\linewidth]{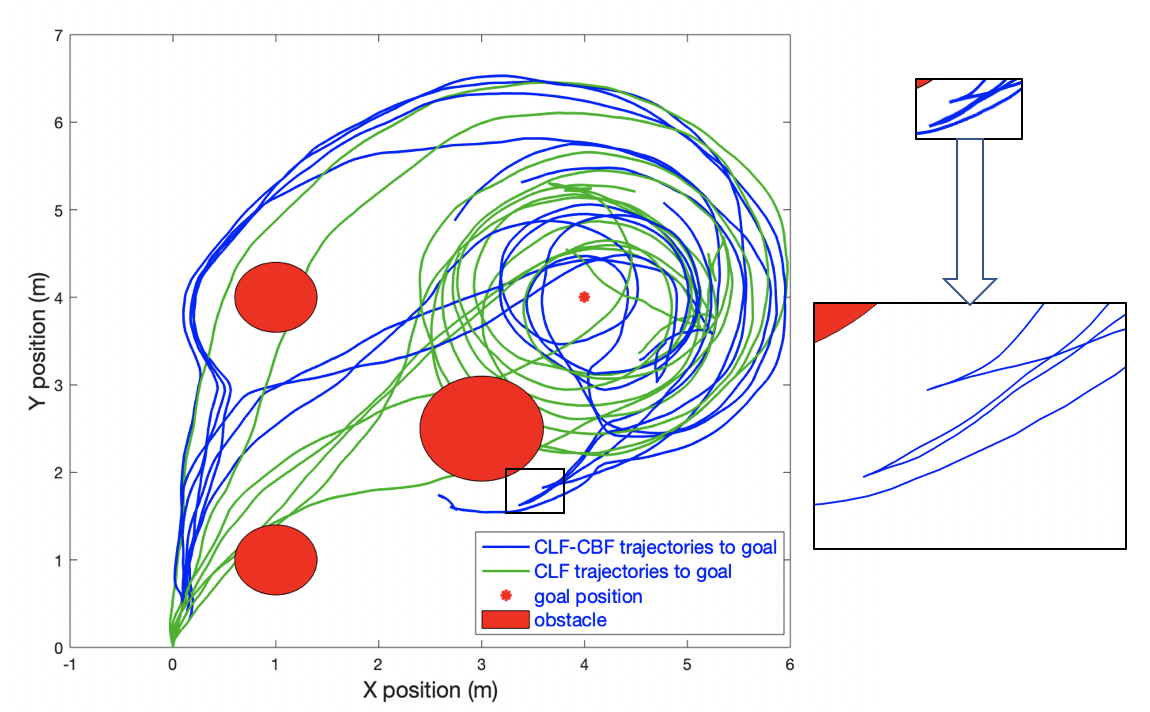}
  \caption{Multiple obstacle avoidance task: 7 test trajectories out of the 20 sample trajectories are shown here. The green ones generated with the CLF controller show no regard for the obstacles where as the blue ones from the CLF-CBF controller successfully avoids all the 3 obstacles in all cases. The selected enlarged portion shows how the blue trajectories takes a \textit{U-turn} when they encountered the obstacle shown in red. }
\label{fig:2DCar_traj_multi_obstacle}
\end{figure}
\subsection{2-Link Pendulum with Elastic Actuators (Relative Deg. 4)}
Here we consider a two link pendulum with elastic actuator described in  \cite{nguyen2016exponential}. The two-link pendulum is a system with 4 degrees of freedom and two degrees of underactuation. The commanded torques $\tau_1$ and $\tau_2$ of the two motors in the two elastic actuators will generate torque at the two joints indirectly through the following motor dynamics:
\begin{align*}
    J_m\ddot{\theta}_1^m= k(\theta_1-\theta_1^m) + \tau_1\\
    J_m\ddot{\theta}_2^m= k(\theta_2-\theta_2^m) + \tau_2
\end{align*}
where, $\theta_1^m$ and $\theta_2^m$ are the angles of the motors and $\theta_1^m$ and $\theta_2^m$ are the joint angles as shown in Fig. 2 of \cite{nguyen2016exponential}. $k$ and $J_m$ are the stiffness and inertia of the motor respectively. The Equations for torque at the two joints are given as:
\begin{align*}
    J_1\ddot{\theta}_1=u_1= -k(\theta_1-\theta_1^m) - \xi \dot{\theta}_1\\
    J_2\ddot{\theta}_2=u_1= -k(\theta_2-\theta_2^m) - \xi \dot{\theta}_2
\end{align*}
where $J_1$ and $J_2$ $(J_1>J_2)$ are the inertia of the links and $\xi$ denotes the damping coefficient at the joints. The objective is to swing the pendulum from the initial state $(\theta_1,\theta_2)= (-\pi/2, 0)$ to the final state of $(\theta_1,\theta_2)= (\pi/2, 0)$. We also restrict $-\pi\leq\theta_1\leq \pi$.The control Lyapunov function was defined as :
\begin{equation*}
    V_0(\vx)= \frac{1}{2}\big((\theta_1-\theta_{1desired})^2+(\theta_2-\theta_{2desired})^2\big)
\end{equation*}
The barrier function and the safe set are defined as follows:
\begin{align*}
   \psi_0(\vx) &= h(\vx)= (\theta_{1limit}^2-\theta_1^2)\\
   \mathcal{C}_{safe} &= \{\vx : \psi_0(\vx)\ge 0 \}\\
\end{align*}
We simulated 40 sample trajectories for a time horizon of 60s. $\sigma= 0.05$, $d= 1000$ and $Q$ was taken as Identity. The results from the simulations are presented in figure \ref{fig:2link_theta1_theta2} and figure \ref{fig:2Link_theta1_traj}.From our experiments we found that since the proposed method essentially solves a \ac{QP}, it is highly dependent on how efficiently the optimizer can converge to a feasible solution. The highlighted portion in figure \ref{fig:2Link_theta1_traj} shows one such cases where the optimizer fails to converge to a feasible solution and thus the trajectories violates the safety constraints. 
\begin{figure}
  \includegraphics[width=\linewidth ]{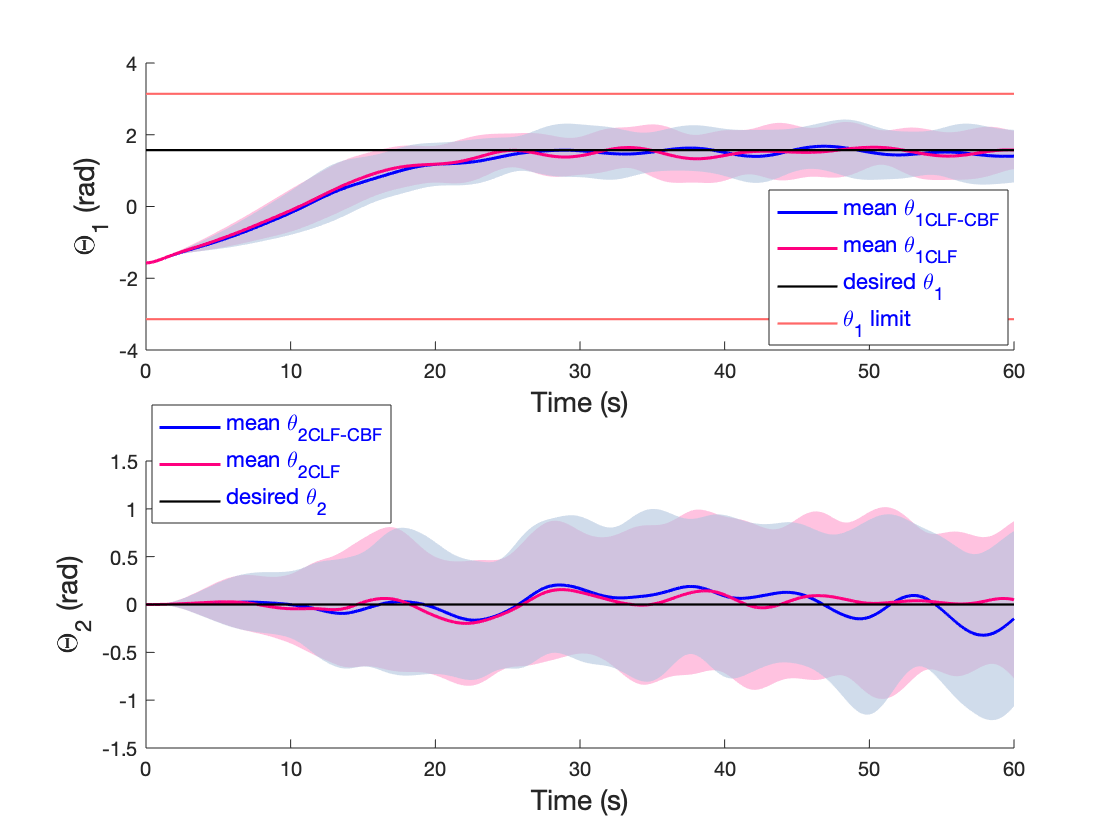}
  \caption{2-Link pendulum task: The mean and the sigma distribution of the $\theta_1$ and $\theta_2$ state trajectories for the CLF-CBF controller and the CLF controller are plotted here in sheds of blue and green respectively.}
\label{fig:2link_theta1_theta2}
\end{figure}
\begin{figure}[H]
  \includegraphics[width=\linewidth]{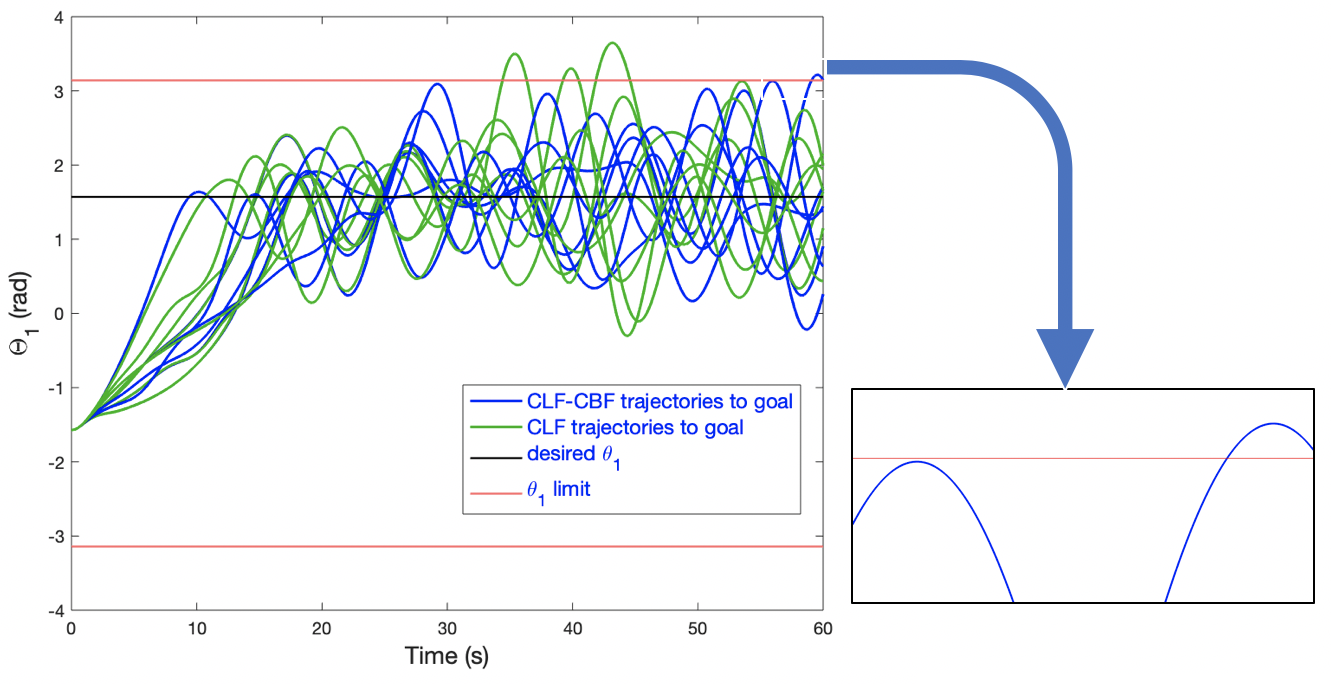}
  \caption{2-Link pendulum task: 7 test trajectories out of the 40 sample trajectories  for $\theta_1$ are shown here. The green ones generated with the CLF controller violates the constraint conditions where as the blue ones from the CLF-CBF controller successfully stays between the $-\pi\leq\theta_1\leq \pi$ bounds. The highlighted case shows the scenario where the optimizer fails to find a feasible solution and violates the boundary conditions. }
\label{fig:2Link_theta1_traj}
\end{figure}
\section{Conclusion and Future Work}\label{sec:conclusion}
We have introduced a generalised framework for control Lyapunov and barrier functions for stochastic systems with high relative degree. We have provided the necessary condition to guarantee asymptotic stability in the sense of probability to stochastic systems with observable states that are not directly actuated via control. Our high relative degree \ac{SCBF} framework ensure safety bounds to a system by providing a means to incorporate positional state dependent constraints into the control problem formulation. Our control policies are evaluated in the form of a \ac{QP} which has state dependent constraints based on \ac{SCLF}-\ac{SCBF} conditions. We have demonstrated the effectiveness of our proposed control methodology with numerical studies on a relative degree 2 system (2D car navigation with multiple obstacles) and 4 system (2-link pendulum with elastic actuator swing task).
\par The present work opens up numerous application of stochastic CLF and CBF formulation to safety critical systems. This formulation will be of particular interest to explore safe and sample efficient Reinforcement Learning frameworks that can incorporate the safety bounds in the formulation of the optimisation problem itself rather than learning them from making mistakes. 




\bibliography{root}
\balance
\bibliographystyle{ieeetr}



%

\end{document}

%% file: packages.tex
\usepackage{hyperref}
\usepackage{url}

\usepackage[utf8]{inputenc} 
\usepackage[T1]{fontenc}    
\usepackage{hyperref}       
\usepackage{url}            
\usepackage{booktabs}       
\usepackage{amsfonts}       
\usepackage{nicefrac}       
\usepackage{microtype}      
\usepackage{algorithmic}
\usepackage{balance}
\usepackage{amsmath}
\usepackage{amssymb}
\usepackage{amsthm}
\usepackage{mathrsfs}
\usepackage{subfigure}
\usepackage{graphicx}
\usepackage{caption}
\usepackage{tabularx}
\usepackage[ruled,linesnumbered]{algorithm2e}
\usepackage{algorithmic}
\usepackage[nolist]{acronym}
\usepackage{bbm}
\usepackage{xcolor}
\usepackage{wrapfig}
\usepackage{float}
\usepackage{siunitx}

%% file: commands.tex
\DeclareMathOperator{\tr}{\rm{tr}}

\newcommand{\rd}{{\mathrm d}}

\newcommand{\calU}{{\cal{U}}}

\newcommand{\Zeta}{\mathit{Z}}

\newcommand{\va}{{\bf a}}

\newcommand{\vx}{{\bf x}}

\newcommand{\vf}{{\bf f}}

\newcommand{\vG}{{\bf G}}
\newcommand{\vu}{{\bf u}}

\newcommand{\vw}{{\bf w}}

\newcommand{\vSigma}{{\mbox{\boldmath$\Sigma$}}}

\newcommand{\T}{^\mathrm{T}}

\newtheorem{theorem}{Theorem}

\newcommand{\Rb}{\mathbb{R}}

\begin{acronym}
\acro{ML}{Machine Learning}
\acro{DNN}{Deep Neural Network}
\acro{DP}{Dynamic Programming}
\acro{FBSDE}{Forward-Backward Stochastic Differential Equation}
\acro{FSDE}{Forward Stochastic Differential Equation}
\acro{BSDE}{Backward Stochastic Differential Equation}
\acro{LSTM}{Long-Short Term Memory}
\acro{FC}{Fully Connected}
\acro{DDP}{Differential Dynamic Programming}
\acro{HJB}{Hamilton-Jacobi-Bellman}
\acro{PDE}{Partial Differential Equation}
\acro{PI}{Path Integral}
\acro{NN}{Neural Network}
\acro{GPs}{Gaussian Processes}
\acro{SOC}{Stochastic Optimal Control}
\acro{RL}{Reinforcement Learning}
\acro{MPOC}{Model Predictive Optimal Control}
\acro{IL}{Imitation Learning}
\acro{RNN}{Recurrent Neural Network}
\acro{FNN}{Feed-forward Neural Network}
\acro{Single FNN}{Single Feed-forward Neural Network}
\acro{DL}{Deep Learning}
\acro{SGD}{Stochastic Gradient Descent}
\acro{SDE}{Stochastic Differential Equation}
\acro{2BSDE}{second-order BSDE}
\acro{CBF}{Control Barrier Function}
\acro{SCBF}{Stochastic Control Barrier Function}
\acro{SCLF}{Stochastic Control Lyapunov Function}
\acro{HDSCLF}{High-Relative Degree Stochastic Control Lyapunov Function}
\acro{QP}{Quadratic Program}
\acro{CBF}{Control Barrier Function}
\acro{PDIPM}{Primal-Dual Interior Point Method}
\end{acronym}

%% file: root.bbl
\begin{thebibliography}{10}

\bibitem{nguyen2016exponential}
Q.~Nguyen and K.~Sreenath, ``Exponential control barrier functions for
  enforcing high relative-degree safety-critical constraints,'' in {\em 2016
  American Control Conference (ACC)}, pp.~322--328, IEEE, 2016.

\bibitem{althoffACM2011}
M.~Althoff, C.~Le~Guernic, and B.~H. Krogh, ``Reachable set computation for
  uncertain time-varying linear systems,'' in {\em Proceedings of the 14th
  International Conference on Hybrid Systems: Computation and Control}, HSCC
  ’11, (New York, NY, USA), p.~93–102, Association for Computing Machinery,
  2011.

\bibitem{mitra2013}
S.~{Mitra}, T.~{Wongpiromsarn}, and R.~M. {Murray}, ``Verifying cyber-physical
  interactions in safety-critical systems,'' {\em IEEE Security Privacy},
  vol.~11, no.~4, pp.~28--37, 2013.

\bibitem{ames2014control}
A.~D. Ames, J.~W. Grizzle, and P.~Tabuada, ``Control barrier function based
  quadratic programs with application to adaptive cruise control,'' in {\em
  53rd IEEE Conference on Decision and Control}, pp.~6271--6278, IEEE, 2014.

\bibitem{ames2019control}
A.~D. Ames, S.~Coogan, M.~Egerstedt, G.~Notomista, K.~Sreenath, and P.~Tabuada,
  ``Control barrier functions: Theory and applications,'' in {\em 2019 18th
  European Control Conference (ECC)}, pp.~3420--3431, IEEE, 2019.

\bibitem{clark2019control}
A.~Clark, ``Control barrier functions for complete and incomplete information
  stochastic systems,'' in {\em 2019 American Control Conference (ACC)},
  pp.~2928--2935, IEEE, 2019.

\bibitem{xiao2019}
W.~Xiao and C.~Belta, ``Control barrier functions for systems with high
  relative degree,'' {\em CoRR}, vol.~abs/1903.04706, 2019.

\bibitem{deng1997}
H.~Deng and M.~Krsti{\'c}, ``Stochastic nonlinear stabilization—ii: inverse
  optimality,'' {\em Systems \& control letters}, vol.~32, no.~3, pp.~151--159,
  1997.

\bibitem{nishimura2013stochastic}
Y.~Nishimura, K.~Tanaka, Y.~Wakasa, and Y.~Yamashita, ``Stochastic asymptotic
  stabilizers for deterministic input-affine systems based on stochastic
  control lyapunov functions,'' {\em IEICE Transactions on Fundamentals of
  Electronics, Communications and Computer Sciences}, vol.~96, no.~8,
  pp.~1695--1702, 2013.

\bibitem{wang2009performance}
Y.~Wang and S.~Boyd, ``Performance bounds for linear stochastic control,'' {\em
  Systems \& Control Letters}, vol.~58, no.~3, pp.~178--182, 2009.

\bibitem{stochastic_stability_book}
R.~Khasminskii, {\em Stochastic stability of differential equations}, vol.~66.
\newblock Springer Science \& Business Media, 2011.

\bibitem{fan2019bayesian}
D.~D. Fan, J.~Nguyen, R.~Thakker, N.~Alatur, A.-a. Agha-mohammadi, and E.~A.
  Theodorou, ``Bayesian learning-based adaptive control for safety critical
  systems,'' {\em arXiv preprint arXiv:1910.02325}, 2019.

\bibitem{todorov2005generalized}
E.~Todorov and W.~Li, ``A generalized iterative lqg method for locally-optimal
  feedback control of constrained nonlinear stochastic systems,'' in {\em
  Proceedings of the 2005, American Control Conference, 2005.}, pp.~300--306,
  IEEE, 2005.

\bibitem{pereira2019learning}
M.~A. Pereira, Z.~Wang, I.~Exarchos, and E.~A. Theodorou, ``Learning deep
  stochastic optimal control policies using forward-backward sdes,'' in {\em
  Robotics: science and systems}, 2019.

\bibitem{exarchos2018stochastic}
I.~Exarchos and E.~A. Theodorou, ``Stochastic optimal control via forward and
  backward stochastic differential equations and importance sampling,'' {\em
  Automatica}, vol.~87, pp.~159--165, 2018.

\bibitem{xie2017}
Z.~{Xie}, C.~K. {Liu}, and K.~{Hauser}, ``Differential dynamic programming with
  nonlinear constraints,'' in {\em 2017 IEEE International Conference on
  Robotics and Automation (ICRA)}, pp.~695--702, 2017.

\end{thebibliography}
